\documentclass[runningheads,envcountsame]{llncs}
\usepackage[hscale=.66, vscale=.82]{geometry}
\usepackage{mysty}

\usepackage{thmtools,thm-restate}
                    
\begin{document}

\title{Automata on Graph Alphabets}

\author{%
  Hugo Bazille\inst1 \and
  Uli Fahrenberg\inst2
}

\institute{%
  EPITA Research Lab (LRE), France \and
  LMF, Université Paris-Saclay
}

\maketitle

\begin{abstract}
The theory of finite automata concerns itself with words in a free monoid together with concatenation and without further structure. There are, however, important applications which use alphabets which are structured in some sense.

We introduce automata over a particular type of structured data, namely an alphabet which is given as a (finite or infinite) directed graph. This constrains concatenation: two strings may only be concatenated if the end vertex of the first is equal to the start vertex of the second.

We develop the beginnings of an automata theory for languages on graph alphabets. We show that they admit a Kleene theorem, relating rational and regular languages, and a Myhill-Nerode theorem, stating that languages are regular iff they have finite prefix or, equivalently, suffix quotient.
We present determinization and minimization algorithms, but we also exhibit that regular languages are not stable by complementation.

Finally, we mention how these structures could be generalized to presimplicial alphabets, where languages are no more freely generated.

\end{abstract}


\newpage

\section{Introduction and Motivation}

The theory of finite automata concerns itself with words over an alphabet $\Sigma$,
that is, the free monoid on $\Sigma$ together with concatenation and without further structure.
There are, however, important applications which use alphabets which are structured in some sense,
for example in database theory~\cite{DBLP:conf/dagstuhl/2005P5061},
model checking of systems with data~\cite{DBLP:journals/siglog/DemriQ21, DBLP:conf/mfcs/PetelerQ22},
type checking~\cite{DBLP:conf/padl/Newton25},
or concurrency theory~\cite{DBLP:journals/tcs/AmraneBFZ25}.

In this paper we work with a particular type of structured data,
namely an alphabet which is given as a (finite or infinite) directed graph.
This constrains concatenation:
two strings may only be concatenated if the end vertex of the first
is equal to the start vertex of the second.
We give a few examples.

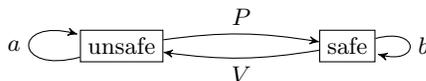
\begin{figure}[bp]
  \centering
  \begin{tikzpicture}[x=2cm]
    \node[state, rectangle, inner sep=1mm] (s) at (0,0) {unsafe};
    \node[state, rectangle, inner sep=1mm] (u) at (1.5,0) {safe};
    \path[loop left] (s) edge node {$a$} (s);
    \path[loop right] (u) edge node {$b$} (u);
    \path[bend left=3mm] (s) edge node {$P$} (u);
    \path[bend left=3mm] (u) edge node {$V$} (s);
  \end{tikzpicture}
  \caption{Alphabet for locking and releasing a resource.}
  \label{fig:lock}
\end{figure}

In a concurrent setting, processes' behavior depends on whether or not they have access to shared resources.
Figure~\ref{fig:lock} shows a simple graph alphabet designed to reflect this,
where the action available depends on the state:
in unsafe state, processes may do $a$,
but after locking the resource ($P$), they can do $b$,
until the resource is released ($V$).
Note that the graph in Fig.~\ref{fig:lock} only constrains the \emph{alphabet};
processes themselves may have arbitrarily complex behavior,
but they are by design restricted to not do $b$ while in unsafe state.

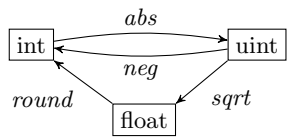
\begin{figure}[bp]
  \centering
  \begin{tikzpicture}[x=1.5cm]
    \node[state, rectangle, inner sep=1mm] (int) at (0,0) {int};
    \node[state, rectangle, inner sep=1mm] (uint) at (2,0) {uint};
    \node[state, rectangle, inner sep=1mm] (float) at (1,-1) {float};
    \path[bend left=3mm] (int) edge node {\textit{abs}} (uint);
    \path[bend left=3mm] (uint) edge node {\textit{neg}} (int);
    \path (uint.south west) edge node {\textit{sqrt}} (float.north east);
    \path (float.north west) edge node {\textit{round}} (int.south east);
  \end{tikzpicture}
  \caption{Some function types.}
  \label{fig:types}
\end{figure}

Another example is inspired by Newton~\cite{DBLP:conf/padl/Newton25}
who uses (symbolic) automata for checking sequences of types.
When integrating function types into this setting, it is necessary to restrict their appearance in sequences,
see Fig.~\ref{fig:types} for a simple example.
Using automata on such graph alphabet will automatically ensure that sequences such as
$\textit{abs\/}\mathord;\textit{neg\/}\mathord;\textit{sqrt}$ do not appear.

Our final example comes from non-interleaving concurrency theory.
\cite{DBLP:conf/RelMiCS/AmraneBCFZ24, DBLP:journals/tcs/AmraneBFZ25}
introduce ST-automata, see Fig.~\ref{fig:sta} for an example.
These are state-labeled automata on an alphabet of so-called starters and terminators,
vectors which indicate events that are started ($a\ibullet$),
events that are terminated ($\ibullet a$),
and these that continue running ($\ibullet a\ibullet$).
Here, a sequence such as $\loset{b\ibullet\\a\ibullet} \loset{\ibullet c\\\ibullet a}$
would be malformed, as there is no event $c$ running which can be terminated by the second element.
ST-automata are best understood as automata on an infinite graph alphabet
where vertices denote events which are running and edges are starters and terminators;
we will get back to this example later.

\begin{figure}[tbp]
  \centering
  \begin{tikzpicture}[x=.7cm, y=.62cm]
    \begin{scope}[shift={(7,-4.3)}, x=0.9cm, y=0.9cm]
      \node[state,draw=black,fill=blue!20,inner sep=0pt,minimum size=15pt]
      (ac) at (0,4) {$\vphantom{hy}\emptyset$};
      \node[state,draw=black,fill=blue!20,inner sep=0pt,minimum size=15pt]
      (cc) at (4,4) {$\vphantom{hy}\emptyset$};
      \node[state,draw=black,fill=blue!20,inner sep=0pt,minimum size=15pt]
      (ae) at (0,8) {$\vphantom{hy}\emptyset$};
      \node[state,draw=black,fill=blue!20,inner sep=0pt,minimum size=15pt]
      (ec) at (8,4) {$\vphantom{hy}\emptyset$};
      \node[state,draw=black,fill=blue!20,inner sep=0pt,minimum size=15pt]
      (ce) at (4,8) {$\vphantom{hy}\emptyset$};
      \node[state,draw=black,fill=blue!20,inner sep=0pt,minimum size=15pt]
      (ee) at (8,8) {$\vphantom{hy}\emptyset$};
      \node[state,draw=black,fill=green!30,inner sep=0pt,minimum size=15pt]
      (bc) at (2,4) {$\vphantom{hy}b$};
      \node[state,draw=black,fill=green!30,inner sep=0pt,minimum size=15pt]
      (ad) at (0,6) {$\vphantom{hy}a$};
      \node[state,draw=black,fill=green!30,inner sep=0pt,minimum size=15pt]
      (be) at (2,8) {$\vphantom{hy}b$};
      \node[state,draw=black,fill=green!30,inner sep=0pt,minimum size=15pt]
      (cd) at (4,6) {$\vphantom{hy}a$};
      \node[state,draw=black,fill=green!30,inner sep=0pt,minimum size=15pt]
      (de) at (6,8) {$\vphantom{hy}c$};
      \node[state,draw=black,fill=green!30,inner sep=0pt,minimum size=15pt]
      (dc) at (6,4) {$\vphantom{hy}c$};
      \node[state,draw=black,fill=green!30,inner sep=0pt,minimum size=15pt]
      (ed) at (8,6) {$\vphantom{hy}a$};
      \node[state,draw=black,fill=black!20,inner sep=0pt,minimum size=15pt]
      (bd) at (2,6) {$\vphantom{hy}\loset{b \\ a}$};
      \node[state,draw=black,fill=black!20,inner sep=0pt,minimum size=15pt]
      (dd) at (6,6) {$\vphantom{hy}\loset{c \\ a}$};
      \path (ac) edge node[below] {$b\ibullet$} (bc);
      \path (bc) edge node[below] {$\ibullet b$} (cc);
      \path (ac) edge node[left] {$a \ibullet$} (ad);
      \path (ad) edge node[left] {$\ibullet a$} (ae);
      \path (bc) edge node[right] {$\loset{\ibullet b \ibullet \\ \hphantom{\ibullet} a \ibullet }$} (bd);
      \path (bd) edge node[left] {$\loset{\ibullet b \ibullet \\ \ibullet a \hphantom{\ibullet}}$} (be);
      \path (cc) edge node[below] {$c\ibullet$} (dc);
      \path (dc) edge node[below] {$\ibullet c$} (ec);
      \path (ad) edge node[above] {$\loset{\hphantom{\ibullet} b \ibullet \\ \ibullet a \ibullet}$} (bd);
      \path (bd) edge node[above] {$\loset{\ibullet b \hphantom{\ibullet}\\ \ibullet a \ibullet}$} (cd);
      \path (ae) edge node[above] {$b \ibullet$} (be);
      \path (be) edge node[above] {$\ibullet b$} (ce);
      \path (bd) edge node[above left=-0.15cm] {$\loset{\ibullet b \\ \ibullet a}$} (ce);
      \path (ac) edge node[above left=-0.15cm] {$\loset{b \ibullet \\ a \ibullet}$} (bd);
      \path (cd) edge node[above] {$\loset{\hphantom{\ibullet}c \ibullet \\ \ibullet a \ibullet}$} (dd);
      \path (dd) edge node[above] {$\loset{\ibullet c \hphantom{\ibullet}\\ \ibullet a \ibullet}$} (ed);
      \path (ce) edge node[above] {$c \ibullet $} (de);
      \path (de) edge node[above] {$\ibullet c$} (ee);
      \path (cc) edge node[left] {$a \ibullet$} (cd);
      \path (cd) edge node[left] {$\ibullet a$} (ce);
      \path (dc) edge node[right] {$\loset{\ibullet c \ibullet \\ \hphantom{\ibullet} a \ibullet}$} (dd);
      \path (dd) edge node[left] {$\loset{\ibullet c \ibullet \\   \ibullet a \hphantom{\ibullet}}$} (de);
      \path (ec) edge node[right]{$a \ibullet $} (ed);
      \path (ed) edge node[right] {$\ibullet a$} (ee);
      \path (dd) edge node[above left=-0.15cm] {$\loset{ \ibullet c \\ \ibullet a }$} (ee);
      \path (cc) edge node[above left=-0.15cm] {$\loset{c \ibullet \\ a \ibullet}$} (dd);
      \node[below left] at (ac) {$\bot\;\;$};
      \node[below right] at (ec) {$\;\;\top$};
    \end{scope}
  \end{tikzpicture}
  \caption{ST-automaton accepting $a\mathop{\|} b c$.}
  \label{fig:sta}
\end{figure}

We develop the beginnings of an automata theory for languages on graph alphabets.
We show that they admit a Kleene theorem, relating rational and regular languages,
and a Myhill-Nerode theorem, stating that languages are regular
iff they have finite prefix or, equivalently, suffix quotient
(subject to an easy condition on the graph).
Further, automata on graph alphabets are determinizable and complementable,
and minimal deterministic automata exist.
Some of the more elementary proofs of our results have been relegated to appendix.

\section{Automata on Graph Alphabet}
\label{se:graut}

A \emph{graph alphabet} $(V, \Sigma, d_0, d_1)$
consists of sets $V$ and $\Sigma$ together with source and target mappings $d_0, d_1: \Sigma\to V$.
We will generally omit $d_0$ and $d_1$ from the signature.
Both $V$ and $\Sigma$ may be finite or infinite,
as may the automata we define next.

\begin{definition}
  \label{de:aut}
  An \emph{automaton} on graph alphabet $(V, \Sigma)$
  is a structure $A=(Q, I,$ $F, E, s, t, \mu, \lambda)$
  consisting of a set of states $Q$ with initial and accepting states $I, F\subseteq Q$,
  a set of transitions $E$ with source and target mappings $s, t: E\to Q$,
  and labelings $\mu: Q\to V$, $\lambda: E\to \Sigma$.
  We require that $\mu(s(e))=d_0(\lambda(e))$ and $\mu(t(e))=d_1(\lambda(e))$ for every $e\in E$.
\end{definition}

That is, the transitions of $A$ are labeled with elements of $\Sigma$,
but in a way consistent with the graph $(V, \Sigma)$:
a transition labeled $a$ must emanate from a state labeled $d_0(a)$ and lead to a state labeled $d_1(a)$.
Motivated by the example in Fig.~\ref{fig:types},
we may thus understand the state labels as \emph{types} which restrict the application of transitions.



A \emph{path} in $A$ is an alternating sequence $\pi=(q_0, e_1, q_1,\dotsc, e_n, q_n)$ of states and transitions
such that $s(e_i)=q_{i-1}$ and $t(e_i)=q_i$ for all $i$;
we naturally expand this notation to $\src(\pi)=q_0$ and $\tgt(\pi)=q_n$.
The \emph{label} of $\pi$ is
\begin{equation*}
  \lambda(\pi)=(\mu(q_0), \lambda(e_1)\dotsm \lambda(e_n), \mu(q_n)).
\end{equation*}
We keep the labels of the start and target state of $\pi$ in $\lambda(\pi)$
as we will need typing information for concatenation.
Note that there are no empty paths in $A$.
The shortest paths are the constant paths $\pi=(q)$ for $q\in Q$;
the label of such a path $\pi$ is $\lambda(\pi)=\mu(q)$.

A path $\pi$ is \emph{accepting} if $\src(\pi)\in I$ and $\tgt(\pi)\in F$.
The \emph{language} of $A$ is $L(A)=\{\lambda(\pi)\mid \pi \text{ accepting path in } A\}$.
Its \emph{untyped language} is $L_b(A)=\{\omega \in \Sigma^*\mid \exists u, v\in V: (u, \omega, v)\in L(A)\}$,
forgetting the typing information.

\begin{proposition}
  \label{pr:bare}
  The untyped language of any finite automaton on graph alphabet is regular.
\end{proposition}

\begin{proof}
  There is a graph homomorphism from any graph alphabet $(V, \Sigma)$ to the one-point graph $(*, \Sigma)$,
  so any automaton on $(V, \Sigma)$ may be mapped to one on $(*, \Sigma)$,
  preserving the untyped language.
  A finite automaton on graph alphabet $(*, \Sigma)$ is just a standard finite automaton on $\Sigma$.
  \qed
\end{proof}

\begin{example}
  \label{ex:sta}
  The ST-automaton of Fig.~\ref{fig:sta} is a finite automaton
  on the graph alphabet of starters and terminators $\Omega(\{a, b, c\})$ on the set $\{a, b, c\}$.
  Figure~\ref{fig:stalph} depicts part of the simpler $\Omega(\{a, b\})$.
  Its vertices are words in $\{a, b\}^*$ (written vertically), hence the graph is infinite.
  Vertices denote labeled events which are active, so for example in $\loset{a\\b}$,
  two events are active, one labeled $a$ and the other $b$.
  (For technical reasons, $\loset{a\\b}\ne \loset{b\\a}$; see~\cite{DBLP:journals/tcs/AmraneBFZ25}.)
  Vertices like $\loset{a\\a}$
  denote \emph{auto-concurrency},
  \ie~several $a$-labeled events in parallel.
  Edges in the graph start or terminate events,
  so a transition like $\loset{\ibullet b\ibullet\\\pibullet a\ibullet}$ from $\loset{b}$ to $\loset{b\\a}$,
  for example, indicates that an $a$-labeled event is started while the $b$-labeled event keeps running.

  Note that the \emph{constraint automata} of
  \cite{DBLP:journals/entcs/Gascon09, DBLP:conf/mfcs/PetelerQ22} have some similarity to our ST-automata:
  also these have transitions labeled by operations of inserting or removing letters from strings;
  but only prefixes and suffixes are considered and there is no state labeling.
  The full subword order is considered for example in \cite{DBLP:conf/lics/HalfonSZ17},
  but we are not aware of any automata-based approach to subword inclusions.
\end{example}

\begin{figure}[tbp]
  \centering
  \begin{tikzpicture}[x=2.8cm, y=2.5cm]
    \node (0) at (0,.5) {$\emptyset$};
    \node (a) at (.7,0) {$\loset{a}$};
    \node (aa) at (1.7,0) {$\loset{a\\a}$};
    \node at (2,0) {$\dotsm$};
    \node (b) at (-.7,0) {$\loset{b}$};
    \node (bb) at (-1.7,0) {$\loset{b\\b}$};
    \node at (-2,0) {$\dotsm$};
    \node (ab) at (-.7,-1) {$\loset{a\\b}$};
    \node at (-1,-1) {$\dotsm$};
    \node (ba) at (.7,-1) {$\loset{b\\a}$};
    \node at (1,-1) {$\dotsm$};

    \path[bend left=3mm] (0) edge node[sloped] {$
      \loset{a\ibullet}
      $} (a);
    \path[bend left=3mm] (a) edge node[sloped, swap] {$
      \loset{\ibullet a}
      $} (0);
    \path[bend left=3mm] (0) edge node[sloped, swap] {$
      \loset{b\ibullet}
      $} (b);
    \path[bend left=3mm] (b) edge node[sloped] {$
      \loset{\ibullet b}
      $} (0);

    \path[bend left=3mm] (a) edge node {$
      \loset{\pibullet a\ibullet\\\ibullet a\ibullet},
      \loset{\ibullet a\ibullet\\\pibullet a\ibullet}
      $} (aa);
    \path[bend left=3mm] (aa) edge node {$
      \loset{\ibullet a\pibullet\\\ibullet a\ibullet},
      \loset{\ibullet a\ibullet\\\ibullet a\pibullet}
      $} (a);

    \path[bend left=2mm] (b) edge node[pos=.2, sloped] {$
      \loset{\ibullet b\ibullet\\\pibullet a\ibullet}
      $} (ba);
    \path[bend left=2mm] (ba) edge node[pos=.2, sloped, swap] {$
      \loset{\ibullet b\ibullet\\\ibullet a\pibullet}
      $} (b);
    \path[bend left=3mm] (b) edge node {$
      \loset{\pibullet a\ibullet\\\ibullet b\ibullet}
      $} (ab);
    \path[bend left=3mm] (ab) edge node {$
      \loset{\ibullet a\pibullet\\\ibullet b\ibullet}
      $} (b);
    \path[bend left=3mm] (b) edge node {$
      \loset{\pibullet b\ibullet\\\ibullet b\ibullet},
      \loset{\ibullet b\ibullet\\\pibullet b\ibullet}
      $} (bb);
    \path[bend left=3mm] (bb) edge node {$
      \loset{\ibullet b\pibullet\\\ibullet b\ibullet},
      \loset{\ibullet b\ibullet\\\ibullet b\pibullet}
      $} (b);

    \path[bend left=2mm] (a) edge node[pos=.8, sloped, swap] {$
      \loset{\ibullet a\ibullet\\\pibullet b\ibullet}
      $} (ab);
    \path[bend left=2mm] (ab) edge node[pos=.8, sloped] {$
      \loset{\ibullet a\ibullet\\\ibullet b\pibullet}
      $} (a);
    \path[bend left=3mm] (a) edge node {$
      \loset{\pibullet b\ibullet\\\ibullet a\ibullet}
      $} (ba);
    \path[bend left=3mm] (ba) edge node {$
      \loset{\ibullet b\pibullet\\\ibullet a\ibullet}
      $} (a);

  \end{tikzpicture}
  \caption{Part of graph alphabet of starters and terminators on $\{a, b\}$.}
  \label{fig:stalph}
\end{figure}

The language of an automaton on graph $(V, \Sigma)$
is a set of morphisms in the free category generated by $(V, \Sigma)$ which we denote $(V, \Sigma)^*$.
(This is in analogy to $\Sigma^*$ being the free monoid on a set $\Sigma$,
a case which embeds into ours as the one-point graph $(*, \Sigma)$.)

We say that a set $X$ of morphisms in $(V, \Sigma)^*$ is \emph{regular}
(or $(V, \Sigma)$-regular if ambiguity may arise)
if there exists a finite automaton $A$ on graph alphabet $(V, \Sigma)$ for which $L(A)=X$.
Note that a regular set of morphisms is not necessarily a subcategory of $(V, \Sigma)^*$,
as it may not be closed under concatenation.
The \emph{universal automaton} on $(V, \Sigma)$ is the graph $(V, \Sigma)$ itself,
with all states initial and accepting.
Its language is $(V, \Sigma)^*$.

\section{Kleene Theorem}
\label{se:kleene}

Fix a graph $(V, \Sigma)$ for the rest of this paper.
We define a notion of rational sets of morphisms in $(V, \Sigma)^*$
and show that rationality and regularity agree.

When we write $X\subseteq (V, \Sigma)^*$, we always understand $X$ to be a set of morphisms.
For a morphism $x=(u, \omega, v)$ we write $d_0(x)=u$ and $d_1(x)=v$.
We extend the notation to sets of morphisms, writing
$d_0(X)=\{d_0(x)\mid x\in X\}$ and $d_1(X)=\{d_1(x)\mid x\in X\}$.

Morphisms $x=(u, \epsilon, v)$ are identities in $(V, \Sigma)^*$ and exist only for $u=v$;
for such $x$, $d_0(x)=d_1(x)=u$.
We use the notation $x=\id_u$ for identities.
This defines an embedding of $V$ into $(V, \Sigma)^*$ (as identities)
which we often will use implicitly below.
Similarly, we will identify edges $a\in \Sigma$ with their images $(d_0(a), a, d_1(a))$ in $(V, \Sigma)^*$.

\begin{definition}
  The \emph{concatenation} of $X, Y\subseteq (V, \Sigma)^*$ is
  \begin{equation*}
    X Y=\{(u, \omega \nu, w)\mid \exists v\in V: (u, \omega, v)\in X, (v, \nu, w)\in Y\}.
  \end{equation*}
\end{definition}

The \emph{Kleene plus} of $X\subseteq (V, \Sigma)^*$ is
\begin{equation*}
  X^+=\bigcup_{n\ge 1} X^n,
\end{equation*}
where $X^1=X$ and $X^{n+1}=X X^n$ as usual.

\begin{remark}
  Here the reader may ask why we define a Kleene plus and not a Kleene \emph{star},
  which would include $X^0$ in the infinite union above.
  The reason is that one would want $X^0$ to be the unit for concatenation,
  which forces to define $X^0=\{\id_u\mid u\in V\}$.
  Now in most applications, the graph $(V, \Sigma)$ has infinitely many vertices;
  so then $X^0$ would be infinite and non-regular (see Lem.~\ref{le:rat->d0fin} below).
  This is similar to the situation in typed Kleene algebra~\cite{DBLP:journals/tocl/Kozen00}.
\end{remark}

The \emph{basic} sets are $\emptyset$ and $\{(d_0(a), a, d_1(a))\}$ for each $a\in \Sigma$.

\begin{definition}
  A set $X\subseteq (V, \Sigma)^*$ is \emph{rational} if
  \begin{itemize}
  \item $X$ is basic,
  \item there are rational sets $Y, Z\subseteq (V, \Sigma)^*$ such that $X=Y\cup Z$,
  \item there are rational sets $Y, Z\subseteq (V, \Sigma)^*$ such that $X=Y Z$, or
  \item there is a rational set $Y\subseteq (V, \Sigma)^*$ such that $X=Y^+$.
  \end{itemize}
\end{definition}

We follow the usual strategy for proving that rational sets are precisely regular sets
and transfer the notions of the above definition to the automaton side.

\begin{definition}
  The \emph{basic} automata on $(V, \Sigma)$ are
  \begin{itemize}
  \item $A_\emptyset=(\emptyset, \emptyset, \emptyset, \emptyset, \emptyset, \emptyset, \emptyset, \emptyset)$, the empty automaton, and
  \item $A_{\{a\}}=(\{d_0(a), d_1(a)\}, \{d_0(a)\}, \{d_1(a)\}, \{a\}, d_0, d_1, \id, \id)$
    for every $a\in \Sigma$.
  \end{itemize}
\end{definition}

\begin{lemma}
  \label{le:kleene.basic}
  For every basic automaton $A_X$, $L(A_X)=X$.
\end{lemma}

\begin{proof}
  By inspection.
  \qed
\end{proof}

\begin{definition}
  The \emph{union} of two automata
  $A_1=(Q_1, I_1, F_1, E_1, s_1, t_1, \mu_1, \lambda_1)$ and
  $A_2=(Q_2, I_2,$ $F_2, E_2, s_2, t_2, \mu_2, \lambda_2)$ on $(V, \Sigma)$,
  with $Q_1\cap Q_2=E_1\cap E_2=\emptyset$, is
  \begin{equation*}
    A_1\cup A_2=(Q_1\cup Q_2, I_1\cup I_2, F_1\cup F_2, E_1\cup E_2, s, t, \mu, \lambda),
  \end{equation*}
  with
  \begin{gather*}
    s(e)=
    \begin{cases}
      s_1(e) &\text{if } e\in E_1, \\
      s_2(e) &\text{if } e\in E_2,
    \end{cases}
    \qquad
    t(e)=
    \begin{cases}
      t_1(e) &\text{if } e\in E_1, \\
      t_2(e) &\text{if } e\in E_2,
    \end{cases}
    \\
    \mu(q)=
    \begin{cases}
      \mu_1(q) &\text{if } q\in Q_1, \\
      \mu_2(q) &\text{if } q\in Q_2,
    \end{cases}
    \qquad
    \lambda(e)=
    \begin{cases}
      \lambda_1(e) &\text{if } e\in E_1, \\
      \lambda_2(e) &\text{if } e\in E_2.
    \end{cases}
  \end{gather*}
\end{definition}

\begin{restatable}{lemma}{kleeneunion}
  \label{le:kleene.union}
  For all automata $A_1$, $A_2$ on $(V, \Sigma)$, $L(A_1\cup A_2)=L(A_1)\cup L(A_2)$.
\end{restatable}

In order to treat concatenation, we need to introduce silent transitions.
An automaton on $(V, \Sigma)$ with silent transitions is $A=(Q, I, F, E, s, t, \mu, \lambda)$ as of Def.~\ref{de:aut},
except that $\lambda: E\to \Sigma\cup V$,
so that transitions can also be labeled with vertices,
with the convention that $d_0(v)=d_1(v)=v$ for $v\in V$.
Hence silent transitions cannot change the type of a state.

\begin{remark}
  Equivalently, we could have added silent transitions to the graph $(V, \Sigma)$ itself:
  adding an identity loop to every vertex in $V$.
  An automaton on the so-augmented $(V, \Sigma)_\tau$ is easily seen to be equivalent
  to an automaton on $(V, \Sigma)$ with silent transitions.
\end{remark}

\begin{lemma}
  For every automaton $A$ on $(V, \Sigma)$ with silent transitions,
  there exists another automaton $A'$ on $(V, \Sigma)$ without silent transitions such that $L(A)=L(A')$.
\end{lemma}

\begin{proof}
  The standard closure construction applies.
  \qed
\end{proof}

\begin{definition}
  The \emph{concatenation} of two automata
  $A_1=(Q_1, I_1, F_1, E_1, s_1, t_1,$ $\mu_1, \lambda_1)$ and
  $A_2=(Q_2, I_2, F_2, E_2, s_2, t_2, \mu_2, \lambda_2)$ on $(V, \Sigma)$,
  with $Q_1\cap Q_2=E_1\cap E_2=\emptyset$, is
  \begin{equation*}
    A_1 A_2=(Q_1\cup Q_2, I_1, F_2, E_1\cup E_2\cup E_\tau, s, t, \mu, \lambda),
  \end{equation*}
  with
  \begin{gather*}
    E_\tau = \{(q_1, v, q_2)\mid q_1\in F_1, q_2\in I_2, \mu_1(q_1)=\mu_2(q_2)=v\}, \\
    s(e)=
    \begin{cases}
      s_1(e) &\text{if } e\in E_1, \\
      s_2(e) &\text{if } e\in E_2, \\
      q_1 &\text{if } e=(q_1, v, q_2)\in E_\tau,
    \end{cases}
    \qquad
    t(e)=
    \begin{cases}
      t_1(e) &\text{if } e\in E_1, \\
      t_2(e) &\text{if } e\in E_2, \\
      q_2 &\text{if } e=(q_1, v, q_2)\in E_\tau,
    \end{cases}
    \\
    \mu(q)=
    \begin{cases}
      \mu_1(q) &\text{if } q\in Q_1, \\
      \mu_2(q) &\text{if } q\in Q_2,
    \end{cases}
    \qquad
    \lambda(e)=
    \begin{cases}
      \lambda_1(e) &\text{if } e\in E_1, \\
      \lambda_2(e) &\text{if } e\in E_2, \\
      v &\text{if } e=(q_1, v, q_2)\in E_\tau.
    \end{cases}
  \end{gather*}
\end{definition}

Note that this is an automaton with silent transitions,
and that we are only connecting accepting and initial states of the same type.

\begin{restatable}{lemma}{kleeneconcat}
  \label{le:kleene.concat}
  For all automata $A_1$, $A_2$ on $(V, \Sigma)$, $L(A_1 A_2)=L(A_1) L(A_2)$.
\end{restatable}

\begin{definition}
  The \emph{Kleene plus} of an automaton
  $A=(Q, I, F, E, s, t, \mu, \lambda)$ on $(V, \Sigma)$ is
  \begin{equation*}
    A^+=(Q, I, F, E\cup E_\tau, s', t', \mu, \lambda'),
  \end{equation*}
  with
  \begin{gather*}
    E_\tau = \{(q_1, v, q_2)\mid q_1\in F, q_2\in I, \mu(q_1)=\mu(q_2)=v\}, \\
    s'(e)=
    \begin{cases}
      s(e) &\text{if } e\in E, \\
      q_1 &\text{if } e=(q_1, v, q_2)\in E_\tau,
    \end{cases}
    \qquad
    t'(e)=
    \begin{cases}
      t(e) &\text{if } e\in E, \\
      q_2 &\text{if } e=(q_1, v, q_2)\in E_\tau,
    \end{cases}
    \\
    \lambda'(e)=
    \begin{cases}
      \lambda(e) &\text{if } e\in E, \\
      v &\text{if } e=(q_1, v, q_2)\in E_\tau.
    \end{cases}
  \end{gather*}
\end{definition}

\begin{restatable}{lemma}{kleenestar}
  \label{le:kleene.star}
  For any automaton $A$ on $(V, \Sigma)$, $L(A^+)=L(A)^+$.
\end{restatable}

We now have all the ingredients to show that rational sets are regular.
For the inverse direction, one last lemma is necessary.

\begin{lemma}
  \label{le:kleene.regrat}
  For any finite automaton $A$ on $(V, \Sigma)$, $L(A)\subseteq (V, \Sigma)^*$ is rational.
\end{lemma}

\begin{proof}
  We follow the usual Brzozowski-McCluskey procedure.
  Write $A=(Q, I, F,$ $E, s, t, \mu, \lambda)$.
  For this proof we assume that $\lambda: E\to 2^{(V, \Sigma)^*}$ labels transitions by sets of morphisms,
  a natural extension of the formalism.
  (This is sometimes called a generalized automaton.)
  
  If $|Q|=0$, then $L(A)=\emptyset$ which is rational.
  If $|Q|=1$, then $L(A)=\emptyset$ unless $Q=I=F$.
  In that latter case, and writing $Q=\{q\}$,
  \begin{equation*}
    L(A)=\{\mu(q)\}\cup \big(\bigcup\nolimits_{e\in E} \lambda(e)\big)^+,
  \end{equation*}
  which is rational as $E$ is finite.

  Assume $|Q|=2$.
  If $|I|=0$ or $|F|=0$, then $L(A)=\emptyset$.
  If $|I|=2$ or $|F|=2$, then we may split $A$ into several copies,
  one for each element of $I\times F$;
  the copies each have precisely one initial and precisely one accepting state,
  and $L(A)$ is the union of the languages of the copies.

  Hence we may assume $|I|=1$ and $|F|=1$.
  Write $Q=\{q_1, q_2\}$.
  If $I=F=\{q_1\}$ (the case $I=F=\{q_2\}$ is symmetric), then
  \begin{equation*}
    L(A)=\{\mu(q_1)\}\cup ((\{\mu(q_1)\}\cup L_{1 1}^+) L_{1 2} (\{\mu(q_2)\}\cup L_{2 2}^+) L_{2 1})^+,
  \end{equation*}
  with
  \begin{equation*}
    L_{i j}=\bigcup \{\lambda(e)\mid e\in E, s(e)=q_i, t(e)=q_j\},
  \end{equation*}
  which hence is rational.
  If $I=\{q_1\}$ and $F=\{q_2\}$ (the other case being symmetric), then
  \begin{equation*}
    L(A)=((\{\mu(q_1)\}\cup L_{1 1}^+) L_{1 2} (\{\mu(q_2)\}\cup L_{2 2}^+) L_{2 1})^+ L_{1 2} (\{\mu(q_2)\}\cup L_{2 2}^+).
  \end{equation*}

  In case $|Q|\ge 3$, we remove a state.
  We may again assume $|I|=|F|=1$, so there is a state $q\in Q\setminus (I\cup F)$.
  Now define $A'=(Q', I, F, E'\cup E_q, s', t', \mu\rest{Q'}, \lambda')$,
  with $Q'=Q\setminus\{q\}$, $E'=\{e\in E\mid s(e), t(e)\ne q\}$, $\mu\rest{Q'}$ the restriction to $Q'$, and
  \begin{gather*}
    \begin{multlined}
      E_q=\big\{\big(
      q_1,
      \lambda(e_1)
      \big(\mu(q)\cup \big(\bigcup \{\lambda(e)\mid e\in E, s(e)=t(e)=q\}\big)^+\big)
      \lambda(e_2),
      q_2\big)
      \bigmid{} \\[-.5ex]
      \qquad q_1, q_2\in Q', e_1, e_2\in E, s(e_1)=q_1, t(e_1)=s(e_2)=q, t(e_2)=q_2\big\},
    \end{multlined} \\
    s'(e)=
    \begin{cases}
      s(e) &\text{if } e\in E, \\
      q_1 &\text{if } e=(q_1, \omega, q_2)\in E_q,
    \end{cases}
    \qquad
    t'(e)=
    \begin{cases}
      t(e) &\text{if } e\in E, \\
      q_2 &\text{if } e=(q_1, \omega, q_2)\in E_q,
    \end{cases}
    \\
    \lambda'(e)=
    \begin{cases}
      \lambda(e) &\text{if } e\in E, \\
      w &\text{if } e=(q_1, \omega, q_2)\in E_q.
    \end{cases}
  \end{gather*}
  One then easily shows that $L(A')=L(A)$.
  \qed
\end{proof}

\begin{theorem}
  A set $X\subseteq (V, \Sigma)^*$ is regular iff it is rational.
\end{theorem}

\begin{proof}
  By Lemmas \ref{le:kleene.basic}, \ref{le:kleene.union}, \ref{le:kleene.concat}, \ref{le:kleene.star}, and \ref{le:kleene.regrat}.
  \qed
\end{proof}

In addition to these operations, we show that languages are stable by intersection.

\begin{definition}
  The \emph{intersection} of two automata
  $A_1=(Q_1, I_1, F_1, E_1, s_1, t_1, \mu_1, \lambda_1)$ and
  $A_2=(Q_2, I_2, F_2, E_2, s_2, t_2, \mu_2, \lambda_2)$ on $(V, \Sigma)$,
  with $Q_1\cap Q_2=E_1\cap E_2=\emptyset$, is
  \begin{equation*}
    A_1\cap A_2=(Q, I, F, E, s, t, \mu, \lambda),
  \end{equation*}
  with
  \begin{gather*}
    Q = \{(q_1, q_2) \in Q_1 \times Q_2 \mid \mu_1(q_1) = \mu_2(q_2)\}
    \\
    I = Q \cap (I_1 \times I_2) \qquad F = Q \cap (F_1 \times F_2)
    \\
    E = \{(e_1, e_2)\in E_1\times E_2\mid \lambda_1(e_1)=\lambda_2(e_2)\}
    \\
    s((e_1, e_2)) = (s_1(e_1), s_2(e_2))
    \qquad
    t((e_1, e_2)) = (t_1(e_1), t_2(e_2))
    \\
    \mu((q_1, q_2))= \mu_1(q_1)
    \qquad
    \lambda((e_1, e_2)) = \lambda_1(e_1)
  \end{gather*}
\end{definition}

\begin{lemma}
  \label{le:kleene.intersection}
  For all automata $A_1$, $A_2$ on $(V, \Sigma)$, $L(A_1\cap A_2)=L(A_1)\cap L(A_2)$.
\end{lemma}


\begin{proof}
  A morphism $x$ is recognized by a path $\pi = ((q_0, q_0'), (e_1, e_1') \dotsc (q_n, q'_n))$
  iff $x$ is recognized by $\pi_1 = (q_0, e_1, \dotsc q_n)$ in $A_1$ and $\pi_2 = (q'_0, e'_1, \dotsc q'_n)$ in $A_2$.
  \qed
\end{proof}

\section{Myhill-Nerode Theorem}
\label{se:mn}

For $X\subseteq (V, \Sigma)^*$ and $w\in (V, \Sigma)^*$ we write
\begin{align*}
  w^{-1} X &= \{u\in (V, \Sigma)^*\mid d_0(u)=d_1(w), w u\in X\}, \\
  X w^{-1} &= \{u\in (V, \Sigma)^*\mid d_1(u)=d_0(w), u w\in X\}
\end{align*}
for the left and right quotients.
Further,
\begin{equation*}
  \suff(X)=\{w^{-1} X\mid w\in (V, \Sigma)^*\}, \qquad
  \pref(X)=\{X w^{-1}\mid w\in (V, \Sigma)^*\}.
\end{equation*}
We aim to show that $X$ is regular iff $\suff(X)$, or equivalently $\pref(X)$, is finite.

\begin{lemma}
  \label{le:reg->suff}
  For any finite automaton $A$ on $(V, \Sigma)$, $\suff(L(A))$ and $\pref(L(A))$ are finite.
\end{lemma}

\begin{proof}
  We show the claim for $\suff(L(A))$; the proof for $\pref(L(A))$ is symmetric.
  Write $A=(Q, I, F, E, s, t, \mu, \lambda)$ and, for $q\in Q$,
  $L_{\to q}=\{\lambda(\pi)\mid \src(\pi)\in I, \tgt(\pi)=q\}$ and
  $L_{q\to}=\{\lambda(\pi)\mid \src(\pi)=q, \tgt(\pi)\in F\}$.
  Now for any $w\in (V, \Sigma)^*$,
  \begin{align*}
    u\in w^{-1} L(A) &\iff
    d_0(u)=d_1(w) \land w u\in L(A) \\
    &\iff
    \exists q\in Q: w\in L_{\to q}, u\in L_{q\to},
  \end{align*}
  hence $w^{-1} L(A)=\bigcup \{L_{q\to}\mid q\in Q, w\in L_{\to q}\}$.

  That is, $\suff(L(A))\subseteq \{\bigcup_{q\in Q'} L_{q\to}\mid Q'\subseteq Q\}$,
  and the latter set is finite.
  \qed
\end{proof}

For the other direction, which we again only show for $\suff(X)$,
we define a version of the standard Nerode congruence.
Fix a set $X\subseteq (V, \Sigma)^*$.

\begin{lemma}
  \label{le:rat->d0fin}
  If $X$ is rational, then $d_0(X)$ is finite.
\end{lemma}

\begin{proof}
  By induction on the definition of rationality.
  \qed
\end{proof}

\begin{definition}
  The \emph{Nerode congruence} induced by $X$ is the equivalence relation $\sim_X$ on $(V, \Sigma)^*$
  given by
  \begin{equation*}
    x\sim_X y \iff x^{-1} X=y^{-1} X.
  \end{equation*}
\end{definition}

We now construct an automaton on $(V, \Sigma)$ that recognizes $X$.
It
will be the $\sim_X$-quotient of a restriction of the unfolding of the universal automaton $(V, \Sigma)$.

\begin{definition}
  \label{de:univ}
  The \emph{universal tree} $U=(Q, I, F, E, s, t, \mu, \lambda)$ on $(V, \Sigma)$
  is the automaton given by
  \begin{gather*}
    Q=(V, \Sigma)^*, \qquad I=V, \qquad F=Q, \\
    E=\{(x, a, x a)\mid x\in Q, a\in \Sigma, d_1(x)=d_0(a)\}, \\
    s((x, a, x a))=x, \qquad t((x, a, x a))=x a, \\
    \mu(x)=d_1(x), \qquad \lambda((x, a, x a))=a.
  \end{gather*}
\end{definition}

The automaton $U$ is the unfolding of $(V, \Sigma)$, but we shall neither prove nor make use of that fact here.

\begin{lemma}
  \label{le:unfolding_source}
  For any accepting path $\pi$ in $U$,
  $\src(\pi)=d_0(\lambda(\pi))$ and $\tgt(\pi)=\lambda(\pi)$.
\end{lemma}

\begin{proof}
  For the first claim, we have $\mu(\src(\pi))=d_0(\lambda(\pi))$ by definition,
  and as $\src(\pi)\in V$, $\mu(\src(\pi))=\src(\pi)$.
  For the second claim, we use induction on the length of $\pi$.
  If $\pi=(q)$ is constant, then $q\in V$,
  so $\lambda(\pi)=\mu(q)=q=\tgt(\pi)$.
  If $\pi=\psi e$ for $e\in E$,
  then $\lambda(\pi)=\lambda(\psi) \lambda(e)$,
  and $\lambda(\psi)=\tgt(\psi)=s(e)$ by induction.
  By construction, $e=(s(e), \lambda(e), s(e) \lambda(e))$,
  hence $\tgt(\pi)= s(e) \lambda(e)=\lambda(\psi) \lambda(e)=\lambda(\pi)$.
  \qed
\end{proof}

\begin{restatable}{lemma}{universaltree}
  $L(U)=(V, \Sigma)^*$.
\end{restatable}

Next we restrict $U$ to recognize $X$.
Define an automaton $U\rest{X}=(Q, I, F, E, s, t, \mu, \lambda)$ on $(V, \Sigma)$ like in Def.~\ref{de:univ},
but with $Q=\{x\in (V, \Sigma)^*\mid x^{-1} X\ne \emptyset\}$, $I=V\cap d_0(X)$, and $F=X$.
Hence the states of $U\rest{X}$ are prefixes of morphisms in $X$,
and $x\in Q$ is accepting iff $x\in X$, \ie~$d_0(x)\in x^{-1} X$.

\begin{restatable}{lemma}{urestx}
  \label{le:urestx}
  $L(U\rest{X})=X$.
\end{restatable}

\begin{proof}
  Let $x \in L(U\rest{X})$.
  There is an accepting path $\pi$ in $U\rest{X}$ with $\lambda(\pi) = x$ and $\tgt(\pi) \in X$.
  By Lemma~\ref{le:unfolding_source}, $\tgt(\pi) = x$.
  Thus, $x \in X$.

  Reciprocally, let $x \in X$.
  Let $x'$ be the biggest prefix of $x$ such that there is a path $\pi'$ in $U\rest{X}$ with $\lambda(\pi') = x'$.
  Suppose that $x' \neq x$ and denote $x = x'ax''$.
  We have that $(x'a)^{-1}X \neq \emptyset$, as it contains $x''$.
  Then $x'a \in Q$ and $(x', a, x'a) \in E$.
  $\pi = \pi'(x',a,x'a)$ is a path in $U\rest{X}$ with $\lambda(\pi)$ a prefix of $x$ bigger than $x'$, a contradiction.
  Hence $x = x'$, and $\tgt(\pi') = \lambda(\pi') = x$ and since $x \in X$
  and  $X = F$, $\pi'$ is accepting and $x \in L(U \rest{x})$.
  \qed
\end{proof}

Finally, we quotient $U\rest{X}$ by $\sim_X$.
We need a lemma that this is well-defined:

\begin{lemma}
  \label{le:mn}
  Let $x, y\in (V, \Sigma)^*$ such that $x\sim_X y$ and $y^{-1} X\ne \emptyset$.
  Then $x^{-1} X\ne \emptyset$, $d_1(x)=d_1(y)$,
  and for all $a\in \Sigma$ with $d_0(a)=d_1(x)$, $x a\sim_X y a$.
\end{lemma}

\begin{proof}
  First, $x^{-1} X=y^{-1} X\ne \emptyset$ as claimed.
  Take any $u\in y^{-1} X$, then $d_1(y)=d_0(u)=d_1(x)$ because $u\in x^{-1} X$.
  (If $y^{-1} X=\emptyset$, this argument would fail.)
  Now let $a\in \Sigma$ with $d_0(a)=d_1(y)$, then
  $(x a)^{-1} X=a^{-1} x^{-1} X=a^{-1} y^{-1} X=(y a)^{-1} X$.
  \qed
\end{proof}

Below, $[{\cdot}]$ denotes equivalence classes of $\sim_X$.

\begin{definition}
  The \emph{Myhill-Nerode automaton} of $X\subseteq (V, \Sigma)^*$
  is\/ $\MN(X)=
  (Q, I, F,$ $E, s, t,$ $\mu, \lambda)$ defined as follows:
  \begin{gather*}
    Q = \{x\in (V, \Sigma)^*\mid x^{-1} X\ne \emptyset\}_{/{\sim_X}} \\
    E = \{([x], a, [x a])\mid [x]\in Q, a\in \Sigma, d_0(a)=d_1(x)\} \\
    I = (V\cap d_0(X))_{/{\sim_X}} \qquad
    F = X_{/{\sim_X}} \\
    s(([x], a, [x a]))=[x] \qquad
    t(([x], a, [x a]))=[x a] \\
    \mu([x])=d_1(x) \qquad
    \lambda(([x], a, [x a]))=a
  \end{gather*}
\end{definition}

$\MN(X)$ is well-defined by Lem.~\ref{le:mn}.

\begin{restatable}{lemma}{mnlang}
  \label{le:mnlang}
  $L(\MN(X))=X$.
\end{restatable}

\begin{lemma}
  \label{le:suff->reg}
  If $(V, \Sigma)$ is such that the subgraph induced by any two vertices is finite,
  and $\suff(X)$ is finite, then so is $\MN(X)$.
\end{lemma}


\begin{proof}
  If $\suff(X)$ is finite then in $\MN(X)$ so is $Q$.
  If $(V, \Sigma)$ is such that the subgraph induced by any two vertices is finite,
  then also $E$ is finite.
  \qed
\end{proof}

\begin{example}
  Let $(V, \Sigma)$ be the graph with $V=\{u, v\}$ and $\Sigma=\{(u, a, v)\mid a\in \Nat\}$,
  and $X=(V, \Sigma)^*$.  Then $X=\Sigma$ and $\MN(X)=U_{(V, \Sigma)}$ which is infinite.
  This shows that the condition on $(V, \Sigma)$ in Lem.~\ref{le:suff->reg} is necessary.
  It is satisfied, for example, for the infinite ST-graph of Ex.~\ref{ex:sta}.
\end{example}

\begin{theorem}
  If $(V, \Sigma)$ is such that the subgraph induced by any two vertices is finite, then
  a set $X\subseteq (V, \Sigma)^*$ is regular iff $\suff(X)$ is finite, iff $\pref(X)$ is finite.
\end{theorem}

\begin{proof}
  By Lemmas \ref{le:reg->suff}, \ref{le:mnlang}, and \ref{le:suff->reg}.
  \qed
\end{proof}

\section{Determinism}
\label{se:determinism}

\begin{definition}
  An automaton $A=(Q, I, F, E, s, t, \mu, \lambda)$ on $(V, \Sigma)$ is \emph{deterministic}
  if\/ $|\{q\in I\mid \mu(q)=v\}|\le 1$ for each $v\in V$,
  and for all $q\in Q$ and $a\in \Sigma$,
  $|\{e\in E\mid s(e)=q, \lambda(e)=a\}|\le 1$.
\end{definition}

That is, every state has at most one outgoing transition per label,
and there is at most one initial state per type.
The latter is needed as replacement of the usual demand that $|I|=1$,
as each morphism $w\in L(A)$ needs an initial state of type $d_0(w)$ in order to be accepted.

\begin{restatable}{lemma}{det}
  \label{le:det}
  The automaton $U$ of Sect.~\ref{se:mn} is deterministic,
  as are the automata $U\rest{X}$ and $\MN(X)$ for every $X\subseteq (V, \Sigma)^*$.
\end{restatable}

We have shown that automata over $(V, \Sigma)$ are determinizable.
Specializing to the finite case, we get

\begin{theorem}
  If $(V, \Sigma)$ is such that the subgraph induced by any two vertices is finite, then
  for any regular set $X\subseteq (V, \Sigma)^*$
  there is a deterministic finite automaton over $(V, \Sigma)$ which recognizes $X$.
\end{theorem}

\begin{proof}
  By Lem.~\ref{le:det}, $\MN(X)$ is deterministic; $X$ being regular, it is also finite.
  \qed
\end{proof}

We now give another proof of the above theorem,
showing that a variation of the standard subset construction applies in our setting.
This allows us to bypass the condition on $(V, \Sigma)$.

\begin{definition}
  Let $A=(Q, I, F, E, s, t, \mu, \lambda)$ be an automaton over $(V, \Sigma)$.
  The \emph{subset automaton} of $A$ is $2^A=(Q', I', F', E', s', t', \mu', \lambda')$ given by
  \begin{gather*}
    Q'=\{R\subseteq Q\mid \forall q, r\in R: \mu(q)=\mu(r)\}, \\
    I'=\{R\in Q'\mid R\subseteq I\}, \qquad
    F'=\{R\in Q'\mid R\cap F\ne \emptyset\}, \\
    E'=\{(R, a, S)\mid R, S\in Q', S = \{q\mid \exists e\in E: s(e)\in R, \lambda(e)=a, t(e) = q\}\} \\
    s'((R, a, S))=R, \qquad t'((R, a, S))=S, \\
    \{\mu'(R)\}=\{\mu(q)\mid q\in R\}, \qquad
    \lambda'((R, a, S))=a.
  \end{gather*}
\end{definition}

That is, we collect the states of $Q$ into subsets of the same type;
the definition of $\mu'(R)$ above makes sense precisely because $\{\mu(q)\mid q\in R\}$
is a one-element set for every $R\in Q'$.
Given that we only add edges for labels which are present in $A$,
$2^A$ is finite if $A$ is, even in case
$(V, \Sigma)$ admits infinite two-vertex induced subgraphs.

\begin{restatable}{lemma}{propdet}
  $2^A$ is deterministic and $L(2^A)=L(A)$.
  If $A$ is finite, then so is $2^A$.
\end{restatable}

\begin{corollary}
  For any regular set $X\subseteq (V, \Sigma)^*$
  there is a deterministic finite automaton over $(V, \Sigma)$ which recognizes $X$.
\end{corollary}

\section{Complementation}
\label{se:complementation}

We show that automata on $(V, \Sigma)$ can be completed.
If $(V, \Sigma)$ is infinite, then for obvious reasons regular languages over $(V, \Sigma)$ are not stable by complementation.
Otherwise, we show that the standard complementation procedure can be adapted.

\begin{definition}
 An automaton $A=(Q, I, F, E, s, t, \mu, \lambda)$ over $(V, \Sigma)$
 is said to be \emph{complete} if for all $v\in V$ there is $q \in I$ with $\mu(q) = v$,
 and for all $q \in Q$ and $a \in \Sigma$ with $d_0(a) = \mu(q)$, there exists $e \in E$ with $s(e) = q$ and $\lambda(e) = a$.
\end{definition}

\begin{definition}
  The \emph{completion} of an automaton $A=(Q, I, F, E, s, t, \mu, \lambda)$ over $(V, \Sigma)$
  is $A' = (Q', I', F', E', s', t', \mu', \lambda')$ with
  \begin{gather*}
    Q' = Q \cup \{q_{i,v} \mid \forall q \in I, \mu(q_i) \neq v\} \cup \{q_{\bot,v} \mid v \in V \}, \\
    I' = I \cup \{q_{i,v} \mid \forall q \in I, \mu(q_i) \neq v\}, \qquad
    F' = F, \\
    E' = E \cup E_\bot \textnormal{ with } E_\bot = \{(q, a, q_{\bot, v}) \mid d_1(a) = v, q \in Q', \forall e \in E, s(e) \neq q \land \lambda(e) \neq a\}, \\
    s'(e) =
    \begin{cases}
      s(e) \textnormal{ if } e \in E, \\
      q \textnormal{ if } e = (q, a, q'), \\
    \end{cases}
    \qquad
    t'(e) =
    \begin{cases}
      t(e) \textnormal{ if } e \in E, \\
      q' \textnormal{ if } e = (q, a, q'), \\
    \end{cases}\\
    \mu'(q) =
    \begin{cases}
      \mu(q) \textnormal{ if } q \in Q, \\
      v \textnormal{ if } q = q_{i,v}, \\
      v \textnormal{ if } q = q_{\bot,v}, \\
    \end{cases}
    \qquad
    \lambda'(e) =
    \begin{cases}
      \lambda(e) \textnormal{ if } e \in E, \\
      a \textnormal{ if } e = (q, a, q'). \\
    \end{cases}
  \end{gather*}
\end{definition}

\begin{lemma}
  \label{le:completion_coaccessible}
  Let $A'$ be the completion of $A$.
  Any state in $Q' \setminus Q$ is not co-accessible.
  Any accepting path $\pi$ in $A'$ is accepting in $A$, with $\lambda'(\pi) = \lambda(\pi)$.
\end{lemma}

%
%

\begin{restatable}{theorem}{completion}
  Any automaton $A$ on $(V, \Sigma)$ may be completed into $A'$
  with $L(A') = L(A)$.
\end{restatable}

Now that we have complete automata, we can complement them to obtain complements of languages.
For a language $X$ over $(V, \Sigma)$, we denote $X^C = (V, \Sigma)^* \setminus X$ the complement of $X$.

\begin{definition}
 Let $A = (Q, I, F, E, s, t, \mu, \lambda)$ be a complete graph automaton.
 The complementation of $A$ is $A^C = (Q, I, Q \setminus F, E, s, t, \mu, \lambda)$.
\end{definition}

\begin{restatable}{theorem}{complement}
  $L(A^C) = L(A)^C$.
\end{restatable}

\section{Minimization}
\label{se:minimization}

We show that for any automaton $A$ on $(V, \Sigma)$,
$\MN(L(A))$ is the smallest deterministic automaton that recognizes $L(A)$.

\begin{definition}
  A \emph{morphism} $\phi$ between two deterministic automata
  $A_i = (Q_i, I_i, F_i, E_i, s_i, t_i,$ $\lambda_i, \mu_i)$ on $(V, \Sigma)$
  is $\phi : Q_1 \to Q_2$ such that
  \begin{itemize}
  \item $q \in I_1\implies \phi(q) \in I_2$ and $q \in F_1\implies \phi(q) \in F_2$,
  \item $\mu_1(q) = \mu_2(\phi(q))$,
  \item for all $e$ with $s_1(e) = q, t_1(e) = q', \lambda(e) = a$, there exists $e' \in E_2$ with $s_2(e') = \phi(q), t_2(e') = \phi(q'), \lambda(e') = a$.
  \end{itemize}
\end{definition}

\begin{lemma}
 Let $A_1$, $A_2$ such that $L(A_1) = L(A_2)$.
 There exists a morphism $\phi$ from $A_1$ to $A_2$.
\end{lemma}

\begin{theorem}
  Let $A$ be a deterministic automaton on $(V, \Sigma)$.
  Any morphism from\/ $\MN(L(A))$ to $A$ is injective.
\end{theorem}

\begin{proof}
 Let $\MN(L(A)) = (Q, I, F, E, s, t, \lambda, \mu)$ and $A = (Q', I', F', E', s', t', \lambda', \mu')$.
 Let $\phi$ be a morphism from $\MN(L(A))$ to $A$, $[x], [y] \in Q$ with $\phi([x]) = \phi([y])$ and $\pi$ a path in $A$ from an initial state to $\phi([x])$ in $A$.
 Then, $x^{-1}L(A) = \lambda'(\pi)^{-1}L(A) = y^{-1}L(A)$ and $[x] = [y]$.
 \qed
\end{proof}

If $A$ is finite, then the above implies that the number of states of $\MN(L(A))$ is minimal
among deterministic automata recognizing $L(A)$.

\section{Automata on Presimplicial Alphabets}

In this final section we give an outlook on a possible extension of our setting.
It is motivated by Ex.~\ref{ex:sta}, because
the notion of ST-automaton most commonly in use is not in fact based on the graph alphabet in Fig.~\ref{fig:stalph}
but rather on a version which includes \emph{compositions} of starters with starters and of terminators with terminators,
see~\cite{DBLP:journals/corr/abs-2601.17537} for a discussion.
In Fig.~\ref{fig:stalph} that would induce new edges between $\emptyset$ and $\loset{a\\b}$,
labeled $\loset{a\ibullet\\b\ibullet}$ for the starter
and $\loset{\ibullet a\\\ibullet b}$ for the terminator,
and similarly between $\emptyset$ and each of $\loset{a\\a}$, $\loset{b\\a}$, and $\loset{b\\b}$.

Now this structure $(V, \Sigma)$ should not be understood as a graph,
because then $(V, \Sigma)^*$ is \emph{not free}:
some compositions in $(V, \Sigma)^*$ are not freely defined but induced from $(V, \Sigma)$.
Similar concerns would arise with the type-checking example of Fig.~\ref{fig:types}.

Recall that a \emph{presimplicial set}~\cite{books/Dold95}
is a presheaf on the category $\Delta$ of totally ordered finite sets and monotone injections.
That is, $\Gamma\in \Set^{\Delta^\op}$ consists of simplices of different dimensions
(vertices, edges, triangles, etc.) together with face maps which determine how these are glued.

There is a functor $N: \Cat\to \Set^{\Delta^\op}$ which maps (small) categories to their \emph{nerves}:
objects are mapped to vertices, morphisms to edges, pairs of composable morphisms to triangles, etc.
Using the embeddings $\Delta\hookrightarrow \Cat$ and $\Delta\hookrightarrow \Set^{\Delta^\op}$ (Yoneda),
one sees that $N$ has a left adjoint $F: \Set^{\Delta^\op}\to \Cat$, often called \emph{realization}.
Using an analogy with what we did in Sect.~\ref{se:graut},
we prefer to think of $F$ as yielding the \emph{free} category on a presimplicial set $\Gamma$:
morphisms of $F(\Gamma)$ are sequences of compatible edges in $\Gamma$,
subject to an equivalence relation $\sim_\Gamma$
which identifies concatenations with their ``result'' in $\Gamma$ when it exists.

\begin{figure}[tbp]
  \centering
  \begin{tikzpicture}
    \begin{scope}[x=1.7cm]
      \path[fill=gray!50] (0,0) -- (1,0) -- (1,1) -- (0,0);
      \node[state] (0) at (0,0) {};
      \node[state] (1) at (1,0) {};
      \node[state] (2) at (1,1) {};
      \node[state] (3) at (2,1) {};
      \path (0) edge node[swap] {$a$} (1);
      \path (1) edge node[swap] {$b$} (2);
      \path (2) edge node {$c$} (3);
      \path (0) edge node {$d$} (2);
      \path (1) edge node[swap] {$e$} (3);
    \end{scope}
  \end{tikzpicture}
  \caption{A simple presimplicial alphabet.}
  \label{fig:presimp}
\end{figure}
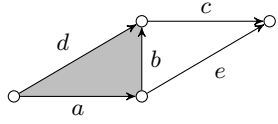

We show a simple example in Fig.~\ref{fig:presimp}.
The equivalence relation induced on $(V, \Sigma)^*$ identifies
$a b\sim_\Gamma d$ and $a b c\sim_\Gamma d c$.
Also note that one-dimensional presimplicial sets are graphs,
and for such the constructions given here degenerate to the ones in Sect.~\ref{se:graut}.

Now let $\Sigma$ be a presimplicial set which will act as our alphabet.
An \emph{automaton} $A$ on $\Sigma$ is defined as in Def.~\ref{de:aut},
with states labeled by points and transitions labeled by edges of $\Sigma$;
but now we have information on compositions in $\Sigma$ which reflect the definition of the language of $A$:
\begin{equation*}
  L(A)=\{x\in F(\Sigma)\mid \exists \pi \text{ accepting path in } A: \lambda(\pi)\sim_\Sigma x\}
\end{equation*}

We conjecture that the proofs in Sect.~\ref{se:kleene} also go through in this generalized setting
and hence that automata on presimplicial alphabets admit a Kleene theorem.
Whether or not the same holds for a Myhill-Nerode theorem and determinization is more doubtful,
given that, for example, ST-automata are related to higher-dimensional automata~\cite{DBLP:journals/tcs/AmraneBFZ25}
which are not generally determinizable~\cite{DBLP:journals/fuin/FahrenbergZ24}.


\section{Conclusion and Further Work}

We have developed the beginnings of an automata theory for automata on graph alphabets:
alphabets which constrain concatenation in that two strings may only be concatenated
if the end vertex of the first is the start vertex of the second.
We have shown Kleene and Myhill-Nerode theorems for such automata,
and further results on determinizability and complementability.

An interesting next step would be to look into MSO-like logics for our automata
and try to prove a Büchi-Elgot-\!Trakhtenbrot theorem.
Also notions of aperiodicity and a corresponding first-order logic are something that one might want to develop.

We are more interested, however, in extensions of the basic formalism.
We have already mentioned automata on presimplicial sets as a possible extension,
motivated by applications in non-interleaving concurrency theory.
Another extension would be to consider graphs with a refinement order on vertices,
so that any action available on a vertex $v$ would also be available on any other of which $v$ is a refinement.
That would allow a more faithful modeling of our first example,
where one should consider that unsafe states are refinements of safe states.

Finally, motivated by our second example,
one might want to consider graph alphabets where the graph has a monadic structure,
or even a Cartesian closed one,
so that product types or higher-order function types would be naturally available.
The monadic structure is also present in ST-automata, even though we haven't used it here;
but for function types the higher-order automata of~\cite{DBLP:conf/lics/Mellies17} seem better suited.


\bibliographystyle{plain}
\bibliography{mybib}

\newpage

\section*{Appendix}

\section*{Proofs of Sect.~\ref{se:kleene}}
  \kleeneunion*
 \begin{proof}
  Let $x$ in $L(A_1) \cup L(A_2)$.
  Wlog, let us suppose that $x \in L(A_1)$
  There exists an accepting path $\pi$ in $A_1$ with $\lambda_1(\pi) = x$.
  By construction, $\pi$ is an accepting path in $A_1 \cup A_2$ with $\lambda(\pi) = \lambda_1(\pi) = x$.
  Hence $x \in L(A_1 \cup A_2)$.

  Reciprocally, let $x \in L(A_1 \cup A_2)$.
  There exists an accepting path $\pi = (q_0, e_1, q_1, \dotsc, e_n, q_n)$ in $A_1 \cup A_2$ such that $\lambda(\pi) = x$.
  Wlog, let us suppose that $q_0 \in I_1$.
  Then by a trivial induction, for all $1 \le i \le n$ $q_i \in Q_1$.
  Thus $\pi$ is an accepting path in $A_1$ and $x \in L(A_1) \cup L(A_2)$.
  \qed
\end{proof}

\kleeneconcat*

\begin{proof}
 Let $x \in L(A_1) L(A_2)$.
 There exist $x_1 \in L(A_1), x_2\in L(A_2)$ and $\pi_1$, $\pi_2$ such that
 $x=x_1 x_2$,
 $\lambda_1(\pi_1) = x_1$, and
 $\lambda_2(\pi_2) = x_2$.
 Then, denoting $e = (\tgt(\pi_1), \mu(\tgt(\pi_1)), \src(\pi_2))$, $\pi = \pi_1 e \pi_2$ is a path in $A_1A_2$ with $\lambda(\pi)=x$, $\src(\pi) \in I_1$, $\tgt(\pi) \in F_2$.
 Thus $x \in L(A_1A_2)$.

 Reciprocally, let $x \in L(A_1 A_2)$.
 There exists a path $\pi=(q_0, e_1, q_1,\dotsc, e_n, q_n)$ with $\lambda(\pi) = x$ such that $q_0 \in I_1, q_n \in F_2, \mu(q_0) = d_0(x)$, and $\mu(q_1) = d_1(x)$.
 There is $i$ such that $e_i \in E_\tau$, and for $j \neq i$, $e_j \in E_1$ iff $i < j$.
 Let $\pi_1 = (q_0, e_1, q_1,\dotsc, q_{i-1})$ and $\pi_2 = (q_i, e_{i+1}, \dotsc, q_n)$.
 We have that $q_0 \in I_1, q_{i-1} \in F_1, q_i \in I_2, q_n \in F_2$ and $\mu_1(q_{i-1}) = \mu(q_{i-1}) = \mu(q_i) = \mu_2(q_i)$.
 Since $\lambda_1(\pi_1) = \lambda(\pi_1)$, $\lambda_2(\pi_2) = \lambda(\pi_2)$ and $\lambda(\pi) = \lambda_1(\pi')\lambda_2(\pi'')$, there exists $x_1 \in L(A_1), x_2 \in L(A_2)$ such that  $x = x_1x_2$.
 Thus $x \in L(A_1) L(A_2)$.
 \qed
\end{proof}

\kleenestar*

\begin{proof}
  Let $x \in L(A^+)$.
  There exists $\pi = (q_0, e_1, q_1, \dotsc, e_n, q_n)$ with $\lambda'(\pi) = x$, $\src(\pi)\in I'$ and $\tgt(\pi')\in F'$.
  Let $\Tau = \{i \mid e_i \in E_\tau\}$ and denote $\pi = \pi_0 e_{i_1} \pi_1 \dotsc e_{i_\Tau} \pi_{[\Tau|}$.
  We have that $x = \lambda'(\pi) = \lambda(\pi_0) \dotsc \lambda(\pi_{|\Tau|})$.
  Further, for all $i \leq |\Tau|$, $\src(\pi_i) \in I, \tgt(\pi_i) \in F$
  and any transition in $\pi_i$ is an element of $E$.
  Thus $\pi_i$ is an accepting path in $A$, and $\lambda(\pi_i) \in L(A)$.
  Hence $x \in L(A)^+$.

  Reciprocally, let $x = x_1, \dotsc, x_n \in L(A)^+$ such that for all $i$, $x_i \in L(A)$.
  Let $\pi_i$ be an accepting path in $A$ such that $\lambda(\pi_i) = x_i$.
  We have that $\mu(\tgt(\pi_i)) = d_1(x_i) = d_0(x_{i+1}) = \mu(\src(\pi_{i+1}))$.
  Now, for all $i<n$, let $e'_i = (\tgt(\pi_i), \mu(\tgt(\pi_i)), \src(\pi_{i+1}))$.
  Then $e'_i \in E_\tau$ for all $i$, $\pi_i$ is a path in $A^+$,
  and the concatenation $\pi' = \pi_0 e'_0 \pi_1 \dotsc e'_{n-1} \pi_n$ is well defined as a path in $A^+$.
  Now, $\lambda(\pi') = x$, $\src(\pi') = \src(\pi_0) \in I'$ and $\tgt(\pi') = \tgt(\pi_{|\Tau|}) \in F'$.
  Thus, $x \in L(A^+)$.
  \qed
\end{proof}

\section*{Proofs of Sect.~\ref{se:mn}}

\universaltree*

\begin{proof}
  The inclusion $L(U) \subseteq (V, \Sigma)^*$ is trivial.
  For the other inclusion, let $x \in (V, \Sigma)^*$.
  We show by induction on the length of $x$ that there exists an accepting path $\pi$ in $U$
  such that $\lambda(\pi) = x$.

  If $x = \id_u$, then $\pi=(u)$ is accepting. 
  Then, if $x = x'a$. 
  By hypothesis, there is an accepting path $\pi'$ in $U$ such that $\lambda(\pi') = x'$. 
  By Lem.~\ref{le:unfolding_source}, $\tgt(\pi') = x'$.
  Let $\pi = \pi'\cdot (x', a, x)$.
  We have that $\lambda(\pi) = x$. 
  As $Q = F$, $\pi$ is accepting and $x \in L(U)$.
\end{proof}


\mnlang*

\begin{proof}
  Let $x \in L(\MN(X))$. 
  There exists a path $\pi$ with $\lambda(\pi)=x$ in $\MN(X)$, and let $\tgt(x) =  [y] \in F$.
  We have $x \sim_X y$.
  By Lem.~\ref{le:mn}, $d_1(x) = d_1(y)$, $x^{-1}X = y^{-1}X$ and as $\id_{d_1(y)} \in y^{-1}X$, $x\ \id_{d_1(y)} = x \in X$.

  Reciprocally, let $x \in X$.
  Let $x'$ be the biggest prefix of $x$ such that there is a path $\pi'$ in $\MN(X)$ with $\lambda(\pi') = x'$.
  Suppose that $x' \neq x$ and denote $x = x'ax''$.
  We have that $(x'a)^{-1}X \neq \emptyset$, as it contains $x''$.
  Then $[x'a] \in Q$ and $([x'], a, [x'a]) \in E$.
  $\pi = \pi'([x'],a,[x'a])$ is a path in $\MN(X)$ with $\lambda(\pi)$ a prefix of $x$ bigger than $x'$, a contradiction.
  Hence $x = x'$, and $\tgt(\pi') = [\lambda(\pi')] = [x] \in X = F$.
  Thus $x \in \MN(X)$.
  \qed
\end{proof}

\section*{Proofs of Sect.~\ref{se:determinism}}

\det*

\begin{proof}
  For the claim about $U$ (Def.~\ref{de:univ}), $I=V$, so for any $v\in V$,
  $|\{q\in I\mid \mu(q)=v\}|=|\{q\in I\mid q=v\}|=1$.
  For any $w\in Q$ and $a\in \Sigma$,
  $|\{e\in E\mid s(e)=x, \lambda(e)=a\}=|\{(x, a, x a)\}|=1$.
  Let $X\subseteq (V, \Sigma)^*$,
  then $U\rest{X}$, being a restriction of $U$, is also deterministic.

  For $\MN(X)$, $|\{q\in I\mid \mu(q)=v\}|\le 1$
  because $I$ is a quotient of the set of initial states of $U\rest{X}$.
  For any $[x]\in Q$ and $a\in \Sigma$,
  $|\{e\in E\mid s(e)=[x], \lambda(e)=a\}\le |\{([x], a, [x a])\}|=1$.
  \qed
\end{proof}

\propdet*

\begin{proof}
  Let $(R, a, S), (R, a, S') \in E'$.
  Let $\mathcal s \in S$.
  There is $e \in E, r \in R$, $s(e) = r$ and $t(e) = \mathcal s$.
  Then, $\mathcal s \in S'$ and $S \subseteq S'$.
  Symmetrically, $S' \subseteq S$ and $S = S'$.
  $2^A$ is then deterministic.

  If $A$ is finite, then $|Q'| \le 2^{|Q|}$ and $2^A$ is finite.

  Let $x \in L(2^A)$.
  There exists a path $\pi' = (Q_0, e'_1, Q_1,\dotsc, e'_n, Q_n)$ in $2^A$ that accepts $x$, with $Q_0 \subseteq I$ and $Q_n \cap F \neq \emptyset$.
  We show by induction that for all $i \le n$, for all $q_n \in Q_n$ there is a path $\pi = (q_0, e_1, q_1, \dotsc, q_n)$ such that for all $i \le n$, $q_i \in Q_i$ and $\lambda(e_i) = \lambda'(e'_i)$.
  The case $n = 0$ is trivial.
  If $n > 0$, then $\pi' = \psi' e'$ with $e' = (Q_{n-1}, a, Q_n)$.
  Let $q_n \in Q_n$. There is $q_{n-1} \in Q_{n-1}, e_n \in E$ with $s(e_n) = q_{n-1}$ and $t(e_n) = q_n$.
  By hypothesis there is a path $\psi =  (q_0, e_1, q_1, \dotsc, q_{n-1})$ such that for all $i < n-1$, $q_i \in Q_i$ and $\lambda(e_i) = \lambda'(e'_i)$.
  Then the path $\pi = \psi (q_{n-1}, e_n, q_n)$ exists in $A$ and is labelled by $x$.
  Choosing $q_n \in Q_n \cap F$, we have that $q_0 \in I$ and $\pi$ accepts $x$.

  Reciprocally, let $x \in L(A)$.
  There is a path $\pi = (q_0, e_1, q_1, \dotsc, q_n)$ with $q_0 \in I$, $q_n \in F$ and $\lambda(\pi) = x$.
  We show by induction on the length of $\pi$ that there is $\pi' = (Q_0, e'_1, Q_1, \dotsc, Q_n)$ in $2^A$ with $\lambda'(\pi') = x$ and for all $i \leq n$, $q_i \in Q_i$ and $Q_0 \in I'$.
  The case $n = 0$ is again trivial.
  If $n > 0$, then $\pi = \psi e$.
  By hypothesis there is a is a path $\psi' =  (Q_0, e'_1, Q_1, \dotsc, Q_{n-1})$ such that for all $i \le n-1$, $q_i \in Q_i$ and $\lambda(e_i) = \lambda'(e'_i)$ and $Q_0 \in I'$.
  As there is a transition $(q_{n-1}, a, q_n)$ in $E$, there exists $Q_n$ with $q_n \in Q_n$ and $e'_n = (Q_{n-1}, a, Q_n) \in E'$.
  Then $\pi' = \psi' e'_n$, with have that $\lambda'(\pi') = \lambda(\pi) = x$, $\src(\pi')\in I'$ and $\tgt'(\pi') \in F'$.
  Thus $x \in L(2^A)$.
  \qed
\end{proof}

\section*{Proofs of Sect.~\ref{se:complementation}}

\completion*

\begin{proof}
 Let $x \in L(A)$.
 There exists an accepting path $\pi$ in $A$ with $\lambda(\pi) = x$.
 Then $\pi$ is an accepting path in $A'$, and $x \in L(A')$.

 Reciprocally, let $x \in L(A')$.
 There exists an accepting path  $\pi$ in $A'$ with $\lambda'(\pi) = x$.
 According to Lem.~\ref{le:completion_coaccessible}, $\pi$ is an accepting path in $A$, and $\lambda(\pi) = x$.
 Hence $x \in L(A)$.
 \qed
\end{proof}

\complement*

\begin{proof}
  Let $x \in L(A)^C$.
  As $A$ is complete and deterministic, there exists a unique path $\pi$ in $A$ with $\lambda(\pi) = x$ starting in an initial state, and $t(\pi) \in Q \setminus F$.
  $\pi$ is also an accepting path in $A^C$ labelled with $x$.
  Hence, $x \in L(A^C)$.

  Let $x \in L(A^C)$.
  As $A^C$ is also complete and deterministic, there exists a unique path $\pi$ in $A^C$ with $\lambda(\pi) = x$ starting in an initial state, and $t(\pi) \in Q \setminus F$.
  $\pi$ is also the unique path in $A$ starting in an initial state labelled with $x$, and $t(\pi) \not \in F$.
  Hence, $x \in L(A)^C$.
  \qed
\end{proof}

\end{document}